\documentclass[review]{elsarticle}
\usepackage[letterpaper,top=2cm,bottom=2cm,left=3cm,right=3cm,marginparwidth=1.75cm]{geometry}
\usepackage{lineno,hyperref,mathtools}

\journal{Journal of \LaTeX\ Templates}

\usepackage{amssymb}
\usepackage{bm}
\usepackage{amssymb}
\usepackage{amsthm}
\usepackage{multirow,booktabs}
\usepackage{cases}
\usepackage{threeparttable}

\newtheorem{theorem}{Theorem}

\newtheorem{remark}{Remark}

\newtheorem{lemma}{Lemma}
\newtheorem{corollary}{Corollary}
\newtheorem{example}{Example}

\newcommand{\C}{{\mathcal{C}}}
\newcommand{\F}{{\mathbb{F}}}

\begin{document}

\begin{frontmatter}

\title{New Quantum MDS codes from Hermitian self-orthogonal generalized Reed-Solomon codes}
\tnotetext[mytitlenote]{This research work is supported by the National Natural Science Foundation of China under Grant Nos. U21A20428 and 12171134.}

\author[mymainaddress]{Ruhao Wan}
\ead{wanruhao98@163.com}

\author[mymainaddress]{Shixin Zhu\corref{mycorrespondingauthor}}
\cortext[mycorrespondingauthor]{Corresponding author}
\ead{zhushixinmath@hfut.edu.cn}

\address[mymainaddress]{School of Mathematics, HeFei University of Technology, Hefei 230601, China}

\begin{abstract}
Quantum maximum-distance-separable (MDS for short) codes are an important class of quantum codes.
In this paper,
by using Hermitian self-orthogonal generalized Reed-Solomon (GRS for short) codes,
we construct five new classes of $q$-ary quantum MDS codes with minimum distance larger than $q/2+1$.
Furthermore, the parameters of our quantum MDS code cannot be obtained from the previous constructions.

\end{abstract}


\end{frontmatter}

\section{Introduction}\label{sec1}

Quantum information and quantum computing have become hot topics in recent years.
Quantum codes have important implications for quantum communication and quantum computing.
As in classical coding theory,
a central theme of quantum error correction is the construction of quantum codes with good parameters.
In \cite{RefJ (1998) introduct 1}, Calderbank et al. established a fundamental to use additive codes over $\F_4$ to construct a class of quantum codes named stabilizer codes.
Then Rains \cite{RefJ (1999) introduct 2}, Ashikhmin and Knill \cite{RefJ (2001) introduct 3} generalized their results to general finite fields.
Since then, many quantum codes with good parameters have been constructed by classical linear codes
with certain self-orthogonality (see \cite{RefJ (2006) Keykar}-\cite{RefJ (2004) n<q+1(2)}).

Let $q$ be a prime power.
A $q$-ary quantum code $Q$ of length $n$ and size $K$ is a $K$-dimensional subspace of a $q^n$-dimensional Hilbert space
$\mathbb{H}=(\mathbb{C}^q)^{\bigotimes n}=\mathbb{C}^q\bigotimes \dots \bigotimes \mathbb{C}^q$.
As in classical coding theory, the third important parameter of a quantum code is its minimum distance besides its length and size.
If a quantum code has minimum distance $d$,
then it can detect any $d-1$ quantum errors and correct any $\lfloor \frac{d-1}{2}\rfloor$ quantum errors.
We use the notation $[[n,k,d]]_q$ to denote a quantum code of length $n$, dimension $q^k$ and minimum distance $d$.
One of the main problems of quantum coding theory is to construct quantum codes with the minimum distance as large as possible.
However, there are many tradeoffs between $n$, $k$ and $d$ for a quantum code.
One of the most well-known tradeoffs is the so-called $quantum\ Singleton\ bound$ (see \cite{RefJ (2001) introduct 3,RefJ (2006) Keykar}):
\begin{equation}\label{eq singleton bound}
2d\leq n-k+2.
\end{equation}
The quantum codes that reach the above bound (\ref{eq singleton bound}) are called quantum MDS codes.

\subsection{Related works}

In recent years, the construction of new quantum MDS codes has attracted many researchers, and relevant results have been gradually improved.
More precisely,
all $q$-ary quantum MDS codes of length $n\leq q+1$ have been constructed in \cite{RefJ (2004) n<q+1(1),RefJ (2004) n<q+1(2)}.
As described in \cite{RefJ (2010) L.jin,RefJ (2014) L.jin}, except for some sparse lengths, almost all known $q$-ary quantum MDS codes have a minimum distance less than or equal to $q/2+1$.
So it becomes increasingly challenging to construct some new quantum MDS codes with large minimum distances, especially those larger than $q/2+1$.
Researchers have made efforts to construct quantum MDS codes via negacyclic codes (see \cite{RefJ (2013) n=q^2+1(3)}),
constacyclic codes (see \cite{RefJ (2014) kai}-\cite{RefJ (2015) B.chen}) and pseudo-cyclic codes (see \cite{RefJ (2016) S.Li}).
As an important subclass of MDS codes, GRS codes are also widely used to construct quantum MDS codes.
Li et al. \cite{RefJ (2008) n=q^2+1(4)} first proposed a unified framework for constructing quantum MDS codes via GRS codes.
Then, Jin et al. \cite{RefJ (2010) L.jin,RefJ (2014) L.jin} generalized the method in \cite{RefJ (2008) n=q^2+1(4)}.
Later, using GRS codes, the researchers constructed many quantum MDS codes with minimum distance greater than $q/2+1$ (see \cite{RefJ (2016) X.He}-\cite{RefJ (2021) n=q^2+1(5)}).
In particular, recently, Ball \cite{RefJ (2021) n=q^2+1(5)} proved that the minimum distance of quantum codes derived from GRS codes is at most $q+1$.
We summarize some known results of quantum MDS codes in Table \ref{tab:1} of Section \ref{sec5}.

\subsection{Our results}

In this paper, we focus on the construction of quantum MDS codes via Hermitian self-orthogonal GRS codes.
The main idea of our constructions is to find suitable code locators $a_1,a_2,\dots,a_n\in \F_{q^2}$
and column multipliers $v_1,v_2,\dots,v_n\in \F_{q^2}^*$
such that a system of homogenous equations over $\F_{q^2}$ has a solution over $\F_q^*$ (see Lemma \ref{lem (a,v)=0}).
Firstly, we give some lemmas
(see Lemmas \ref{lem youjie 1}-\ref{lem zhenchu 3}).
Among them,
we mainly apply Lemmas \ref{lem youjie 1}, \ref{lem youjie xin} and \ref{lem youjie xin 2} to prove that a system of equations has solution over $\F_q^*$.
Using these lemmas,
we construct five new classes of Hermitian self-orthogonal GRS codes,
and further draw five classes of quantum MDS codes (see Theorems \ref{th1}, \ref{th2}, \ref{th3}, \ref{th4}, \ref{th5}).
It is noteworthy that the quantum MDS codes constructed in this paper have larger minimum distance than the previous results (see Remarks \ref{remrak 1}, \ref{remrak 2}, \ref{remrak 3}).
Further, the minimum distance of all the $q$-ary quantum MDS codes constructed in this paper can be larger than $q/2+1$.
We list our new constructions of quantum MDS codes in Table \ref{tab:3} of Section \ref{sec5}.

\subsection{Organization of this paper}

The rest of this paper is organized as follows.
In Section \ref{sec2}, we recall some basic results about Hermitian self-orthogonal GRS codes.
In Section \ref{sec3}, we introduce the main lemmas needed in this paper.
In Section \ref{sec4}, we construct five new classes of quantum MDS codes.
In Section \ref{sec5}, we compare the results of this paper with previous results.
In Section \ref{sec6}, we conclude this paper.

\section{Preliminaries}\label{sec2}

In this section, we recall some definitions and basic theories of Hermitian self-orthogonal GRS codes over $\F_{q^2}$
and restate some lemmas.

Let $q$ be a prime power and $\F_{q^2}$ be the finite field with $q^2$ element.
An $[n,k,d]_{q^2}$ linear code $\C$ over $\F_{q^2}$ is just a $k$-dimensional subspace of $\F_{q^2}^n$.
The minimum distance $d$ of $\C$
is equal to the the minimum nonzero Hamming weight of all codewords in $\C$.
For a vector $\bm{c}=(c_1,c_2,\dots,c_n)\in \F_{q^2}^n$,
we always denote a vector $(c_1^i,c_2^i,\dots,c_n^i)\in \F_{q^2}^n$ by $\bm{c}^i$.
Specially, $0^0=1$.
For any two vectors $\bm{x}=(x_1,x_2,\dots,x_n),\ \bm{y}=(y_1,y_2,\dots,y_n)\in \F_{q^2}^n$,
the Euclidean and Hermitian inner product of vectors $\bm{x}$, $\bm{y}$ are defined by
$$\langle \bm{x},\bm{y} \rangle_{E}=\sum_{i=1}^{n}x_iy_i \quad and\quad \langle \bm{x},\bm{y} \rangle_{H}=\sum_{i=1}^{n}x_iy_i^q,$$
respectively.
And the Euclidean and Hermitian dual codes of $\C$ are defined by
$$\C^{\perp_{E}}=\{\bm{x}\in \F_{q^2}^n:\langle \bm{x}, \bm{y}\rangle_{E}=0, \ for \ all\ \bm{y}\in \C \}$$
and
$$\C^{\perp_{H}}=\{\bm{x}\in \F_{q^2}^n:\langle \bm{x}, \bm{y}\rangle_{H}=0, \ for \ all\ \bm{y}\in \C \},$$
respectively.
If $\C\subseteq \C^{\perp_H}$, the code $\C$ is called Hermitian self-orthogonal.
Particularly, if $\C=\C^{\perp_{H}}$,
we call $\C$ a Hermitian self-dual code.
It is easy to check that $C^{\perp_H}=(C^{(q)})^{\perp_E}$, where $C^{(q)}=\{\bm{c}^q:\bm{c}\in C\}$.
For a matrix $A=(a_{ij})$ over $\F_{q^2}$,
we define the matrix $(a_{ij}^q)$ as $A^{(q)}$.

Now, we recall some definitions related to GRS codes.
Choose $\bm{a}=(a_{1},a_{2},\dots,a_{n})$
to be an $n$-tuple of distinct elements of $\F_{q^2}$.
Let $\bm{v}=(v_1,v_2,\dots,v_n)$ be a vector of $(\F_{q^2}^*)^n$, where $\F_{q^2}^*=\F_{q^2}\setminus \{0\}$.
For an integer $k$ satisfying $1\leq k\leq n$, the GRS code associated with $\bm{a}$ and $\bm{v}$ is defined by
$$GRS_{k}(\bm{a},\bm{v})=\{(v_{1}f(a_{1}),v_{2}f(a_{2}),\dots,v_{n}f(a_{n})): f(x)\in \F_{q^2}[x],\ \deg(f(x))\leq k-1\}.$$
It is well known that $GRS_{k}(\bm{a},\bm{v})$ is an $[n,k,n-k+1]_{q^2}$ MDS code.
And its dual code is also MDS.
The elements $a_1, a_2,\dots,a_n$ are called the $code\ locators$ of
$GRS_k(\bm{a},\bm{v})$,
and the elements $v_1,v_2,\dots,v_n$ are called the $column\ multipliers$ of $GRS_k(\bm{a},\bm{v})$.
Note that the vectors $(v_1a_1^j,v_2a_2^j,\dots,v_na_n^j),\ j=0,1,\dots,k-1$ form a basis of $GRS_k(\bm{a},\bm{v})$.
Hence, we can directly obtain the following lemma, which was also given in \cite{RefJ (2014) L.jin,RefJ (2017)Lem GRS}.

\begin{lemma}\label{lem (a,v)=0}
The two vectors $\bm{a}=(a_1,a_2,\dots,a_n)$ and $\bm{v}=(v_1,v_2,\dots,v_n)$
are defined as above.
Then $GRS_k(\bm{a},\bm{v})$ is Hermitian self-orthogonal
if and only if $\langle \bm{a}^{qi+j}, \bm{v}^{q+1} \rangle_E=\sum_{l=1}^nv_la_l^j(v_la_l^i)^q=0$, for all $0\leq i,j \leq k-1$.
\end{lemma}

The following lemma is widely used to construct quantum MDS codes.

\begin{lemma}\label{Lem zhuyao}
(Hermitian construction for quantum MDS codes \cite{RefJ (2001) introduct 3})
If there exists a linear code $\C$ with parameters $[n,k,n-k+1]_{q^2}$ and satisfies $\C\subseteq \C^{\bot_H}$,
then there exists a quantum code with parameters $[[n,n-2k,k+1]]_{q}$.
\end{lemma}

By the next lemma, we can get a new quantum MDS code from a known quantum code.

\begin{lemma}\label{lem propagation rule}
(Propagation Rule \cite{RefJ (2015) n=q^2+1(2)})
If there exists a quantum MDS code with parameters $[[n,n-2k,k+1]]_q$,
then there exists a quantum MDS code with parameters $[[n-1,n-2k+1,k]]_{q}$.
\end{lemma}

\section{Main Lemmas}\label{sec3}

For our construction, we still need some other lemmas. In this section, we list them below.

\begin{lemma}(\cite{RefJ (2019)W. some})\label{lem youjie 1}
Suppose that $A$ is an $(r-1)\times r$ matrix over $\F_{q^2}$, where $r>0$.
If $A$ satisfies the following conditions:
\begin{itemize}
\item[(1)] any $r-1$ columns of $A$ are linearly independent;
\item[(2)]$A^{(q)}$ is row equivalent to $A$,
\end{itemize}
then the following equation
$$A\bm{u}^T=\bm{0}^T$$
has a solution $\bm{u}=(u_1,u_2,\dots,u_{r})\in (\F_q^*)^{r}$.
\end{lemma}

\begin{lemma}\label{lem youjie xin}
Suppose that $A$ is a $t\times n$ matrix over $\F_{q}$, where $1\leq t< n< q+1$.
Let $A_i$ be the $t\times (n-1)$ matrix obtained from $A$ by deleting the $i$-th column.
If $rank(A)=rank(A_1)=\dots=rank(A_n)$,
then the following equation
$$A\bm{u}^T=\bm{0}^T$$
has a solution $\bm{u}=(u_1,u_2,\dots,u_{n})\in (\F_q^*)^{n}$.
\end{lemma}

\begin{proof}
Let $rank(A)=rank(A_1)=\dots=rank(A_n)=r$.
Suppose
$$S=\{\bm{u}\in \F_q^n:A\bm{u}^T=\bm{0}^T\}.$$
Then we have $|S|=q^{n-r}$.
For $1\leq i\leq n$,
define
$$S_i=\{\bm{u}\in \F_q^n:  u_i=0,\ A\bm{u}^T=\bm{0}^T\}\quad and \quad V=\{\bm{u}\in (\F_q^*)^n: A\bm{u}^T=\bm{0}^T\}.$$
Hence
$$V=\overline{S_1\cup S_2 \cup\dots \cup S_n},$$
where $\overline{T}=S-T$ for any subsets $T\subseteq S$.
Note that $S_i=\{\bm{u}\in \F_q^{n-1}: A_i\bm{u}^T=\bm{0}^T\}$, it follows that $|S_i|=q^{n-r-1}$.
By Inclusion-Exclusion principle, we have
$$|V|=|S|-\sum_{i=1}^{n}|S_i|+c,$$
for some $c\geq 0$.
Note that $\bm{0}\in S_1\cap S_2\cap\dots\cap S_n$,
then we can get $c\geq 1$.
It follows that
$$|V|=q^{n-r}-nq^{n-r-1}+c.$$
Obviously $|V|>0$.
Therefore, the equation $A\bm{u}^T=\bm{0}^T$
has a solution $\bm{u}=(u_1,u_2,\dots,u_{n})\in (\F_q^*)^{n}$.
This completes the proof.
\end{proof}

\begin{lemma}\label{lem youjie xin 2}
Suppose that $A$ is a $t\times n$ matrix over $\F_{q^2}$, where $1\leq t\leq n< q+1$.
If the following conditions are met:
\begin{itemize}
\item[(1)] the equation $A\bm{u}^T=\bm{0}^T$ has a solution $\bm{u}=(u_1,u_2,\dots,u_{n})\in (\F_{q^2}^*)^{n}$;
\item[(2)]$A^{(q)}$ is row equivalent to $A$,
\end{itemize}
then the following equation
$$A\bm{u}^T=\bm{0}^T$$
has a solution $\bm{u}=(u_1,u_2,\dots,u_{n})\in (\F_q^*)^{n}$.
\end{lemma}

\begin{proof}
Assume that $\bm{c}\in (\F_{q^2}^*)^{n}$ is a solution of the equation $A\bm{u}^T=\bm{0}^T$.
Since $A^{(q)}$ is row equivalent to $A$,
it is easy to check that $\bm{c}^q\in (\F_{q^2}^*)^{n}$ is also a solution of the equation $A\bm{u}^T=\bm{0}^T$.
Let $w$ be a primitive element of $\F_{q^2}$.
Since $n<q+1$, then there exists $j\in \{1,2,\dots,q+1\}$ such that $\bm{c}+w^{j(q-1)}\bm{c}^q\in (\F_{q^2}^*)^{n}$.
Let $\bm{b}=w^j\bm{c}+w^{jq}\bm{c}^q\in (\F_{q^2}^*)^{n}$,
we can know that $\bm{b}^q=\bm{b}$, then $\bm{b}\in (\F_q^*)^n$.
Note that $A\bm{b}^T= w^jA\bm{c}^T+w^{jq}A(\bm{c}^q)^T=\bm{0}^T$.
Therefore, the equation $A\bm{u}^T=\bm{0}^T$
has a solution $\bm{u}=(u_1,u_2,\dots,u_{n})\in (\F_q^*)^{n}$.
This completes the proof.
\end{proof}

In the following, we introduce a lemma proposed in \cite{RefJ (2019)W. some}.
\begin{lemma}\label{lem zhenchu}
\begin{itemize}
\item[(1)]
Denote by $m=\frac{q^2-1}{2s}$ where $2s\mid (q+1)$.
Suppose $1\leq k\leq (s+t+1)\frac{q+1}{2s}-1$, where $0\leq t\leq s-2$.
Then for any $0\leq i,j\leq k-1$, $qi+j=um$ if and only if
$u\in \{0,s-t,s-t+1,\dots,s+t\}$.
\item[(2)]
Denote by $m=\frac{q^2-1}{2s+1}$ where $(2s+1)\mid (q+1)$.
Suppose $1\leq k\leq (s+t+1)\frac{q+1}{2s+1}-1$, where $0\leq t\leq s-1$.
Then for any $0\leq i,j\leq k-1$, $qi+j=um$ if and only if
$u\in \{0,s-t+1,s-t+2,\dots,s+t\}$.
\end{itemize}
\end{lemma}

Similarly, we have the following lemmas.

\begin{lemma}\label{lem zhenchu 2}
Denote by $m=\frac{q^2-1}{2s}$ where $2s\mid (q+1)$.
Suppose $1\leq k\leq (s+t+2)\frac{q+1}{2s}-2$, where $0\leq t\leq s-2$.
Then for any $0\leq i,j\leq k-1$, $qi+j+q+1-\frac{q+1}{2s}=um$ if and only if
$u\in\{s-t,s-t+1,\dots,s+t+1\}$.
\end{lemma}

\begin{proof}
We suppose that $qi+j+q+1-\frac{q+1}{2s}$ is divisible by $\frac{q^2-1}{2s}$ for some $0\leq i,j\leq k-1$.
Then there exists an integer $u$ such that
$$qi+j+q+1-\frac{q+1}{2s}=u\frac{q^2-1}{2s}.$$
Hence, we have
\[\begin{split} qi+j&=[\frac{u(q+1)}{2s}-2]q+q-1-\frac{u(q+1)}{2s}+\frac{q+1}{2s}\\
                           &=[\frac{u(q+1)}{2s}-2]q+\frac{(2s+1-u)(q+1)}{2s}-2.\\
	\end{split}\]
For $1\leq u\leq 2s$,
it follows that
$$i=\frac{u(q+1)}{2s}-2,\quad j=\frac{(2s+1-u)(q+1)}{2s}-2.$$

If $u\geq s+t+2$, then
$$i=\frac{u(q+1)}{2s}-2\geq (s+t+2)\frac{q+1}{2s}-2\geq k,$$
which contradicts to the fact $i\leq k-1$;

If $u\leq s-t-1$, then
$$j=\frac{(2s+1-u)(q+1)}{2s}-2\geq (s+t+2)\frac{q+1}{2s}-2\geq k,$$
which contradicts to the fact $j\leq k-1$.
Hence, $u\in\{s-t,s-t+1,\dots,s+t+1\}$.
This completes the proof.
\end{proof}

\begin{lemma}\label{lem zhenchu 3}
Denote by $m=\frac{q^2-1}{2s}$ where $2s\mid (q-1)$.
Suppose $1\leq k\leq (s+t)\frac{q-1}{2s}$, where $1\leq t\leq s$.
Then for any $0\leq i,j\leq k-1$, $qi+j+\frac{q+1}{2}=um$ if and only if
$u\in \{1,2,\dots,t-1,s+1,s+2,\dots,s+t-1\}$.
\end{lemma}

\begin{proof}
We suppose that $qi+j+\frac{q+1}{2}$ is divisible by $\frac{q^2-1}{2s}$ for some $0\leq i,j\leq k-1$.
Then there exists an integer $u$ such that
$$qi+j+\frac{q+1}{2}=u\frac{q^2-1}{2s}.$$
Hence, for $1\leq u\leq s$,
we have
\[\begin{split} qi+j&=u\frac{q^2-1}{2s}-\frac{q+1}{2}\\
                           &=[\frac{u(q-1)}{2s}-1]q+\frac{(s+u)(q-1)}{2s}.\\
	\end{split}\]
It follows that
$$i=\frac{u(q-1)}{2s}-1,\quad j=\frac{(s+u)(q-1)}{2s}.$$

If $t\leq u\leq s$, then
$$j=\frac{(s+u)(q-1)}{2s}\geq (s+t)\frac{q-1}{2s}\geq k,$$
which contradicts to the fact $j\leq k-1$;

On other hand, for $s+1\leq u\leq 2s$,
we have
\[\begin{split} qi+j&=u\frac{q^2-1}{2s}-\frac{q+1}{2}\\
                           &=[\frac{u(q-1)}{2s}]q+\frac{(u-s)(q-1)}{2s}-1.\\
	\end{split}\]
It follows that
$$i=\frac{u(q-1)}{2s},\quad j=\frac{(u-s)(q-1)}{2s}-1.$$

If $s+t\leq u\leq 2s$, then
$$i=\frac{u(q-1)}{2s}\geq (s+t)\frac{q-1}{2s}\geq k,$$
which contradicts to the fact $i\leq k-1$.
Hence, $u\in \{1,2,\dots,t-1,s+1,s+2,\dots,s+t-1\}$.
This completes the proof.
\end{proof}

\section{The construction of quantum MDS codes}\label{sec4}

In this section, we construct some new classes of quantum MDS codes via Hermitian self-orthogonal GRS codes.
Let $w$ be a primitive element of $\F_{q^2}$ and $q^2-1=mh$.
Suppose that $\langle \theta\rangle$ is a cyclic subgroup of the multiplicative group $\F_{q^2}^*$, where $\theta=w^h\in \F_{q^2}$.
We can know that $w^{i_1}\theta^{j_1}\neq w^{i_2}\theta^{j_2}$ for any $0\leq i_1\neq i_2\leq r-1$ and $0\leq j_1\neq j_2\leq m-1$, where $1\leq r \leq h$.

\subsection{First construction of quantum MDS codes}

Let $h\mid (q+1)$. Put
$$\bm{a}=(0, w, w\theta,\dots,w\theta^{m-1},\dots,w^r,w^{r}\theta,\dots,w^{r}\theta^{m-1})\in \F_{q^2}^{r\frac{q^2-1}{h}+1},$$
and
$$\bm{v}=(v_0,\underbrace{v_1,\dots,v_1}_{m\ times},\dots,\underbrace{v_r,\dots,v_r}_{m\ times}),$$
where $v_0,\dots,v_r\in \F_{q^2}^*$.
Then when $(i,j)=(0,0)$, we have
\begin{equation}\label{eq qiuhe 0}
\langle \bm{a}^{0},\bm{v}^{q+1}\rangle_E=v_0^{q+1}+m\sum_{l=1}^rv_l^{q+1}.
\end{equation}
And when $(i,j)\neq (0,0)$, we have
$$\langle \bm{a}^{qi+j},\bm{v}^{q+1}\rangle_E=\sum_{l=1}^rw^{l(qi+j)}v_l^{q+1}\sum_{\nu=0}^{m-1}\theta^{\nu(qi+j)},$$
thus
\begin{equation}\label{eq qiuhe 2}
\langle \bm{a}^{qi+j},\bm{v}^{q+1}\rangle_E=\begin{cases}
0; & if\ m\nmid (qi+j),\\
m\sum_{l=1}^rw^{l(qi+j)}v_l^{q+1}; & if\  m\mid (qi+j).
\end{cases}
\end{equation}

\begin{theorem}\label{th1}
Let $n=r\frac{q^2-1}{h}+1$,
where $h\mid (q+1)$, $h\geq 3$ and $1< r< \min\{q,h\}$.
If $2\nmid (r+h)$,
then for any $1\leq k\leq (\frac{r+h-1}{2})\frac{q+1}{h}-1$,
there exists an $[[n,n-2k,k+1]]_q$-quantum MDS code.
\end{theorem}

\begin{proof}
Keep the notations as above. We divide our proof into the following two parts according to whether $h$ is even or odd.

$\bullet$ $\textbf{Case 1:}$ $h=2s$.

Since $2\nmid (r+h)$, suppose $r=2t+3$, where $0\leq t\leq s-2$ and $1\leq k\leq (s+t+1)\frac{q+1}{2s}-1$.
By Lemma \ref{lem zhenchu}, $qi+j=um$ if and only if $u\in \{0,s-t,s-t+1,\dots,s+t\}$.
Then by Eq. (\ref{eq qiuhe 2}), when $qi+j=um$ for $s-t\leq u\leq s+t$,
we have
$$\langle \bm{a}^{qi+j},\bm{v}^{q+1}\rangle_E=m\sum_{l=1}^{r}w^{l u m}v_l^{q+1}. $$
Let $\alpha=w^m\in \F_{q^2}^*$ be a primitive $2s$-th root of unity.
Let $a=s-t$, then we can get $\alpha^{a+v}\neq \alpha^{a+v'}\neq 1$
for any $0\leq v\neq v'\leq r-3$.
Let
$$A=\begin{pmatrix}
 1&   1    &  1  &  \dots  & 1  \\
 0& \alpha^{a} &  \alpha^{2a} & \dots & \alpha^{ra}\\
 0& \alpha^{a+1} &  \alpha^{2(a+1)}  &\dots & \alpha^{r(a+1)}\\
 \vdots& \vdots &  \vdots & \ddots & \vdots\\
 0& \alpha^{a+r-3} & \alpha^{2(a+r-3)} & \dots & \alpha^{r(a+r-3)}\\
\end{pmatrix}$$
be an $(r-1)\times (r+1)$ matrix over $\F_{q^2}$.
For $1\leq i\leq r+1$, let $A_i$ be the $(r-1)\times r$ matrix obtained from $A$ by deleting the $i$-th column.
Then
$$rank(A_1)=rank(B_1)=r-1,$$
where
$$B_1=\begin{pmatrix}
   1    &  1  &  \dots  & 1  \\
   1 &  \alpha^{a} & \dots & \alpha^{(r-1)a}\\
  1 &  \alpha^{a+1}  &\dots & \alpha^{(r-1)(a+1)}\\
   \vdots &  \vdots & \ddots & \vdots\\
 1 & \alpha^{a+r-3} & \dots & \alpha^{(r-1)(a+r-3)}\\
\end{pmatrix}.$$
And for $2\leq i\leq r+1$
$$rank(A_i)=rank(B_i)+1=r-1,$$
where
$$B_i=\begin{pmatrix}
   1           & \dots   &  1              &1        &  \dots  & 1  \\
   \alpha      &\dots   &  \alpha^{i-2}   &\alpha^{i}       & \dots & \alpha^{r}\\
  \alpha^2     &\dots   &  \alpha^{2(i-2)} & \alpha^{2i}     &\dots & \alpha^{2r}\\
   \vdots      &\ddots  &  \vdots             &\vdots   & \ddots    & \vdots\\
 \alpha^{r-3} &\dots    & \alpha^{(r-3)(i-2)} &\alpha^{(r-3)i}           & \dots & \alpha^{(r-3)r}\\
\end{pmatrix}.$$
It follows that $rank(A)=rank(A_1)=\dots=rank(A_{r+1})=r-1$.
By Lemma \ref{lem youjie xin}, equation $A\bm{u}^T=\bm{0}^T$
has a solution $\bm{u}=(u_0,u_1,\dots,u_{r})\in (\F_{q^2}^*)^{r+1}$.
Note that $\alpha^q=\alpha^{-1}$ and $\alpha^{2s}=1$,
we have
$$\alpha^{l(a+j)q}=\alpha^{-l(s-t+j)}=\alpha^{l(s+t-j)}=\alpha^{l(a+2t-j)},$$
for any $1\leq l\leq r$ and $0\leq j\leq 2t$.
It follows that $A$ is row equivalent to $A^{(q)}$.
By Lemma \ref{lem youjie xin 2},
the equation $A\bm{u}^T=\bm{0}^T$
has a solution $\bm{u}=(u_0,u_1,\dots,u_{r})\in (\F_q^*)^{r+1}$.

$\bullet$ $\textbf{Case 2:}$ $h=2s+1$.

Since $2\nmid (r+h)$, suppose $r=2t+2$, where $0\leq t\leq s-1$ and $1\leq k\leq (s+t+1)\frac{q+1}{2s}-1$.
By Lemma \ref{lem zhenchu}, $qi+j=um$ if and only if $u\in \{0,s-t+1,s-t+2,\dots,s+t\}$.
Then by Eq. (\ref{eq qiuhe 2}), when $qi+j=u m$ for $s-t+1\leq u\leq s+t$,
we have
$$\langle \bm{a}^{qi+j},\bm{v}^{q+1}\rangle_E=m\sum_{l=1}^{r}w^{l u m}v_l^{q+1}. $$
Let $\alpha=w^m\in \F_{q^2}^*$ be a primitive $(2s+1)$-th root of unity.
Put $a=s-t+1$.
Let the definition of $A$ be the same as that in $\textbf{Case 1}$.
Similar to the proof of $\textbf{Case 1}$,
the equation $A\bm{u}^T=\bm{0}^T$
has a solution $\bm{u}=(u_0,u_1,\dots,u_{r})\in (\F_q^*)^{r+1}$.

Combining $\textbf{Case 1}$ and $\textbf{Case 2}$,
let $v_l^{q+1}=u_l$ for $1\leq l\leq r$ and let $v_0\in \F_{q^2}^*$ such that $v_0^{q+1}=u_0m$.
By Eq. (\ref{eq qiuhe 0}), we have
$$\langle \bm{a}^{0},\bm{v}^{q+1}\rangle_E=v_0^{q+1}+m\sum_{l=1}^rv_l^{q+1}=m\sum_{l=0}^{r}u_l=0.$$
And by Eq. (\ref{eq qiuhe 2}),
 when $qi+j=um$, we have
$$\langle \bm{a}^{qi+j},\bm{v}^{q+1}\rangle_E=m\sum_{l=1}^{r}w^{l u m}v_l^{q+1}=m\sum_{l=1}^{r}\alpha^{lu}u_l=0.$$
Hence $\langle \bm{a}^{qi+j},\bm{v}^{q+1}\rangle_E=0$, for all $0\leq i,j\leq k-1$.
Therefore, by Lemma \ref{lem (a,v)=0},
the code $GRS_k(\bm{a},\bm{v})$ is Hermitian self-orthogonal.
Then the desired result follows from Lemma \ref{Lem zhuyao}.
This completes the proof.
\end{proof}

Fang et al. [\cite{RefJ (2019)W. some}, Theorem 6.3] constructed a family of quantum MDS codes of length $n=(2t+1)\frac{q^2-1}{2s}$, where $2s\mid (q+1)$.
Applying the propagation rule (see Lemma \ref{lem propagation rule}) for Theorem \ref{th1} Case 1,
we can immediately obtain the following result.

\begin{corollary}
Let $n=r\frac{q^2-1}{2s}$,
where $2s\mid (q+1)$, $r=2t+1$ and $1\leq t< \min\{\frac{q-1}{2},s\}$.
Then for any $1\leq k\leq (s+t)\frac{q+1}{2s}-2$,
there exists an $[[n,n-2k,k+1]]_q$-quantum MDS code.
\end{corollary}

\begin{example}
In this example,
we give some quantum MDS codes from Theorem \ref{th1}.
\begin{itemize}
\item[(1)]
When $8\mid (q+1)$, let $(r,h)=(5,8)$ in Theorem \ref{th1} Case 1.
Then for any $1\leq k\leq \frac{3}{4}(q+1)-1$,
there exists a $[[\frac{5}{8}(q^2-1)+1,\frac{5}{8}(q^2-1)+1-2k,k+1]]_q$-quantum MDS code;
\item[(2)]
When $9\mid (q+1)$, let $(r,h)=(8,9)$ in Theorem \ref{th1} Case 2.
Then for any $1\leq k\leq \frac{8}{9}(q+1)-1$, there exists
a $[[\frac{8}{9}(q^2-1)+1,\frac{8}{9}(q^2-1)+1-2k,k+1]]_q$-quantum MDS code.
\end{itemize}
\end{example}

\subsection{Second construction of quantum MDS codes}

Let $h\mid (q+1)$. Put
$$\bm{a}=(w^{i_1}, w^{i_1}\theta,\dots,w^{i_1}\theta^{m-1},\dots,w^{i_r},w^{i_r}\theta,\dots,w^{i_r}\theta^{m-1})\in \F_{q^2}^{r\frac{q^2-1}{h}},$$
where $i_1,i_2,\dots, i_r$ are distinct modulo $h$, and
$$\bm{v}=(v_1,v_1w^{h-1},\dots,v_1w^{(m-1)(h-1)},\dots,v_r,v_rw^{h-1}\dots,v_rw^{(m-1)(h-1)}),$$
where $v_1,\dots,v_r\in \F_{q^2}^*$.
Then for any $0\leq i,j\leq k-1$, we have
$$\langle \bm{a}^{qi+j},\bm{v}^{q+1}\rangle_E=\sum_{l=1}^rw^{i_l(qi+j)}v_l^{q+1}\sum_{\nu=0}^{m-1}\theta^{\nu(qi+j+q+1-\frac{q+1}{h})},$$
thus
\begin{equation}\label{eq qiuhe 3}
\langle \bm{a}^{qi+j},\bm{v}^{q+1}\rangle_E=\begin{cases}
0; & if\ m\nmid (qi+j+q+1-\frac{q+1}{h}),\\
m\sum_{l=1}^rw^{i_l(qi+j)}v_l^{q+1}; & if\  m\mid (qi+j+q+1-\frac{q+1}{h}).
\end{cases}
\end{equation}

\begin{theorem}\label{th2}
Let $n=r\frac{q^2-1}{h}$,
where $h\mid (q+1)$, $h=2s$, $r=2t+3$ and $0\leq t\leq s-2$.
Then for any $1\leq k\leq (\frac{r+h+1}{2})\frac{q+1}{h}-2$,
there exists an $[[n,n-2k,k+1]]_q$-quantum MDS code.
\end{theorem}

\begin{proof}
Keep the notations as above.
Denote $\alpha=w^m$ and $\xi=w^{\frac{q+1}{h}-q-1}$.
Let $a=s-t$ and let
$$A=\begin{pmatrix}
 \alpha^{i_1a}\xi^{i_1} &  \alpha^{i_2a}\xi^{i_2} & \dots & \alpha^{i_ra}\xi^{i_r}\\
 \alpha^{i_1(a+1)}\xi^{i_1} &  \alpha^{i_2(a+1)}\xi^{i_2}  &\dots & \alpha^{i_r(a+1)}\xi^{i_r}\\
 \vdots &  \vdots & \ddots & \vdots\\
 \alpha^{i_1(a+r-2)}\xi^{i_1} & \alpha^{i_2(a+r-2)}\xi^{i_2} & \dots & \alpha^{i_r(a+r-2)}\xi^{i_r}\\
\end{pmatrix}$$
be an $(r-1)\times r$ matrix over $\F_{q^2}$.
Note that $\xi^q=\alpha\xi$, $\alpha^q=\alpha^{-1}$ and $\alpha^{2s}=1$.
Then we have
\[\begin{split} \alpha^{i_l(a+j)q}\xi^{i_lq}&=(\alpha^{s+t+1-j}\xi)^{i_l}\\
                           &=\alpha^{i_l(a+r-2-j)}\xi^{i_l},\\
	\end{split}\]
for any $1\leq l\leq r$ and $0\leq j\leq r-2$.
It follows that $A$ is row equivalent to $A^{(q)}$.
Let $A_{i}\ (1\leq i\leq r)$ be the $(r-1)\times (r-1)$ matrix obtained from $A$
by deleting the $i$-th column.
Note that $\det(A_{i})$ is equal to a nonzero constant times a Vandermonde determinant.
So $\det(A_{i})\neq 0$.
By Lemma \ref{lem youjie 1},
the equation $A\bm{u}^T=\bm{0}^T$
has a solution $\bm{u}=(u_1,u_2,\dots,u_{r})\in (\F_q^*)^{r}$.
Let $v_l^{q+1}=u_l$ for $1\leq l\leq r$.
Since $1\leq k\leq (s+t+2)\frac{q+1}{2s}-2$,
by Lemma \ref{lem zhenchu 2}, $qi+j+q+1-\frac{q+1}{h}=um$ if and only if $u\in \{s-t,s-t+1,\dots,s+t+1\}$.
Then by Eq. (\ref{eq qiuhe 3}), when $qi+j+q+1-\frac{q+1}{h}=u m$ for $s-t\leq u\leq s+t+1$,
we have
$$\langle \bm{a}^{qi+j},\bm{v}^{q+1}\rangle_E=m\sum_{l=1}^{r}w^{i_l(u m+\frac{q+1}{h}-q-1)}v_l^{q+1}=m\sum_{l=1}^{r}\alpha^{i_lu}\xi^{i_l}u_l=0.$$
Then $\langle \bm{a}^{qi+j},\bm{v}^{q+1}\rangle_E=0$, for all $0\leq i,j\leq k-1$.
Therefore, by Lemma \ref{lem (a,v)=0},
the code $GRS_k(\bm{a},\bm{v})$ is Hermitian self-orthogonal.
Then the desired result follows from Lemma \ref{Lem zhuyao}.
This completes the proof.
\end{proof}

\begin{example}
In this example,
we give some quantum MDS codes from Theorem \ref{th2}.
\begin{itemize}
\item[(1)]
When $8\mid (q+1)$, let $(r,h)=(5,8)$ in Theorem \ref{th2}.
Then for any $1\leq k\leq \frac{7}{8}(q+1)-2$, there exists
a $[[\frac{5}{8}(q^2-1),\frac{5}{8}(q^2-1)-2k,k+1]]_q$-quantum MDS code;
\item[(2)]
When $6\mid (q+1)$, let $(r,h)=(5,6)$ in Theorem \ref{th2}.
Then for any $1\leq k\leq q-1$, there exists
a $[[\frac{5}{6}(q^2-1),\frac{5}{6}(q^2-1)-2k,k+1]]_q$-quantum MDS code.
\end{itemize}
\end{example}

\subsection{Third construction of quantum MDS codes}

Let $q$ be an odd prime power.
Suppose $h=2s$ and $h\mid (q-1)$. Put
$$\bm{a}=(w^{2i_1}, w^{2i_1}\theta,\dots,w^{2i_1}\theta^{m-1},\dots,w^{2{i_r}},w^{2{i_r}}\theta,\dots,w^{2{i_r}}\theta^{m-1})\in \F_{q^2}^{r\frac{q^2-1}{h}},$$
where $i_1,i_2,\dots, i_r$ are distinct modulo $s$,
and
$$\bm{v}=(v_1,v_1w^{s},\dots,v_1w^{(m-1)s},\dots,v_r,v_rw^{s}\dots,v_rw^{(m-1)s}),$$
where $v_1,\dots,v_r\in \F_{q^2}^*$.
Then for any $0\leq i,j\leq k-1$, we have
$$\langle \bm{a}^{qi+j},\bm{v}^{q+1}\rangle_E=\sum_{l=1}^rw^{2i_l(qi+j)}v_l^{q+1}\sum_{\nu=0}^{m-1}\theta^{\nu(qi+j+\frac{q+1}{2})},$$
thus
\begin{equation}\label{eq qiuhe 4}
\langle \bm{a}^{qi+j},\bm{v}^{q+1}\rangle_E=\begin{cases}
0; & if\ m\nmid (qi+j+\frac{q+1}{2}),\\
m\sum_{l=1}^rw^{2i_l(qi+j)}v_l^{q+1}; & if\  m\mid (qi+j+\frac{q+1}{2}).
\end{cases}
\end{equation}

\begin{theorem}\label{th3}
Let $n=r\frac{q^2-1}{h}$, where
$h\mid (q-1)$, $h=2s$ and $1\leq r\leq s$.
Then for any $1\leq k\leq (s+r)\frac{q-1}{2s}$,
there exists an $[[n,n-2k,k+1]]_q$-quantum MDS code.
\end{theorem}

\begin{proof}
Keep the notations as above.
Denote $\alpha=w^{2m}$ and $\xi=w^{-q-1}$. Then $\alpha, \xi\in \F_q$.
Let
$$A=\begin{pmatrix}
 \alpha^{i_1} \xi^{i_1} &  \alpha^{i_2}\xi^{i_2} & \dots & \alpha^{i_r}\xi^{i_r}\\
 \vdots &  \vdots & \ddots & \vdots\\
 \alpha^{i_1(r-1)}\xi^{i_1} & \alpha^{i_2(r-1)}\xi^{i_2} & \dots & \alpha^{i_r(r-1)}\xi^{i_r}\\
  \alpha^{i_1(s+1)} \xi^{i_1} &  \alpha^{i_2(s+1)}\xi^{i_2} & \dots & \alpha^{i_r(s+1)}\xi^{i_r}\\
 \vdots &  \vdots & \ddots & \vdots\\
 \alpha^{i_1(s+r-1)}\xi^{i_1} & \alpha^{i_2(s+r-1)}\xi^{i_2} & \dots & \alpha^{i_r(s+r-1)}\xi^{i_r}\\
\end{pmatrix}$$
be an $2(r-1)\times r$ matrix over $\F_{q}$.
Let
$$B=\begin{pmatrix}
 \alpha^{i_1} \xi^{i_1} &  \alpha^{i_2}\xi^{i_2} & \dots & \alpha^{i_r}\xi^{i_r}\\
 \alpha^{2i_1} \xi^{i_1} &  \alpha^{2i_2}\xi^{i_2} & \dots & \alpha^{2i_r}\xi^{i_r}\\
 \vdots &  \vdots & \ddots & \vdots\\
 \alpha^{i_1(r-1)}\xi^{i_1} & \alpha^{i_2(r-1)}\xi^{i_2} & \dots & \alpha^{i_r(r-1)}\xi^{i_r}\\
\end{pmatrix}$$
be an $(r-1)\times r$ matrix over $\F_{q}$.
Since $\alpha^s=1$,
then the matrix $A$ is row equivalent to the matrix $B$.
Hence, equation $A\bm{u}^T=\bm{0}^T$ has the same solutions as equation $B\bm{u}^T=\bm{0}^T$.
Let $B_{i}\ (1\leq i\leq r)$ be the $(r-1)\times (r-1)$ matrix obtained from $B$
by deleting the $i$-th column.
Note that $\det(B_{i})$ is equal to a nonzero constant times a Vandermonde determinant.
So $\det(B_{i})\neq 0$.
By Lemma \ref{lem youjie 1}, equation $B\bm{u}^T=\bm{0}^T$
has a solution $\bm{u}=(u_1,u_2,\dots,u_{r})\in (\F_q^*)^{r}$.
Then equation $A\bm{u}^T=\bm{0}^T$
has a solution $\bm{u}=(u_1,u_2,\dots,u_{r})\in (\F_q^*)^{r}$.
Let $v_l^{q+1}=u_l$ for $1\leq l\leq r$.
Since $1\leq k\leq (s+r)\frac{q-1}{2s}$,
by Lemma \ref{lem zhenchu 3}, $qi+j+\frac{q+1}{2}=um$ if and only if $u\in \{1,2,\dots,r-1,s+1,s+2,\dots,s+r-1\}$.
Then by Eq. (\ref{eq qiuhe 4}), when $qi+j+\frac{q+1}{2}=u m$ for $u\in \{1,2,\dots,r-1,s+1,s+2,\dots,s+r-1\}$,
we have
$$\langle \bm{a}^{qi+j},\bm{v}^{q+1}\rangle_E=m\sum_{l=1}^{r}w^{2i_l(um-\frac{q+1}{2}) }\xi^{i_l}v_l^{q+1}=m\sum_{l=1}^{r}\alpha^{i_lu }\xi^{i_l}u_l=0.$$
Then $\langle \bm{a}^{qi+j},\bm{v}^{q+1}\rangle_E=0$, for all $0\leq i,j\leq k-1$.
Therefore, by Lemma \ref{lem (a,v)=0},
the code $GRS_k(\bm{a},\bm{v})$ is Hermitian self-orthogonal.
Then the desired result follows from Lemma \ref{Lem zhuyao}.
This completes the proof.
\end{proof}

\begin{example}
In this example,
we give some quantum MDS codes from Theorem \ref{th3}.
\begin{itemize}
\item[(1)]
When $12\mid (q-1)$, let $(r,h)=(5,12)$ in Theorem \ref{th3}.
Then for any $1\leq k\leq \frac{11}{12}(q-1)$, there exists
a $[[\frac{5}{12}(q^2-1),\frac{5}{12}(q^2-1)-2k,k+1]]_q$-quantum MDS code;
\item[(2)]
When $14\mid (q-1)$, let $(r,h)=(5,14)$ in Theorem \ref{th3}.
Then for any $1\leq k\leq \frac{6}{7}(q-1)$, there exists
a $[[\frac{5}{14}(q^2-1),\frac{5}{14}(q^2-1)-2k,k+1]]_q$-quantum MDS code.
\end{itemize}
\end{example}

\subsection{Fourth construction of quantum MDS codes}

Let $q$ be an odd prime power.
Suppose $h=2s$, $h\mid (q-1)$ and $r=r_1+r_2$. Put
\[\begin{split} \bm{a}=(w^{2i_1}, w^{2i_1}&\theta,\dots,w^{2i_1}\theta^{m-1},\dots,w^{2i_{r_1}},w^{2i_{r_1}}\theta,\dots,w^{2i_{r_1}}\theta^{m-1},\\
                           &w^{2j_1+1},w^{2j_1+1}\theta,\dots,w^{2j_1+1}\theta^{m-1},\dots,w^{2j_{r_2}+1},w^{2j_{r_2}+1}\theta,\dots,w^{2j_{r_2}+1}\theta^{m-1})\in \F_{q^2}^{r\frac{q^2-1}{h}},\\
	\end{split}\]
where $s<r<2s$, $i_1,i_2,\dots, i_{r_1}$ are distinct modulo $s$ and $j_1,j_2,\dots, j_{r_2}$ are distinct modulo $s$.
Set
$$\bm{v}=(v_1,v_1w^{s},\dots,v_1w^{(m-1)s},\dots,v_{r},v_{r}w^{s}\dots,v_{r}w^{(m-1)s}),$$
where $v_1,\dots,v_{r}\in \F_{q^2}^*$.
Then for any $0\leq i,j\leq k-1$, we have
$$\langle\bm{a}^{qi+j},\bm{v}^{q+1}\rangle_E=[\sum_{l_1=1}^{r_1}w^{2i_{l_1}(qi+j)}v_{l_1}^{q+1}+\sum_{l_2=1}^{r_2}w^{(2j_{l_2}+1)(qi+j)}v_{r_1+l_2}^{q+1}]\sum_{\nu=0}^{m-1}\theta^{\nu(qi+j+\frac{q+1}{2})},$$
thus
\begin{equation}\label{eq qiuhe 5}
\langle \bm{a}^{qi+j},\bm{v}^{q+1}\rangle_E=\begin{cases}
0; & if\ m\nmid (qi+j+\frac{q+1}{2}),\\
m\sum\limits_{l_1=1}^{r_1}w^{2i_{l_1}(qi+j)}v_{l_1}^{q+1}+m\sum\limits_{l_2=1}^{r_2}w^{(2j_{l_2}+1)(qi+j)}v_{r_1+l_2}^{q+1}; & if\  m\mid (qi+j+\frac{q+1}{2}).
\end{cases}
\end{equation}

\begin{theorem}\label{th4}
Let $n=r\frac{q^2-1}{h}$, where
$h\mid (q-1)$, $h=2s$ and $s< r< 2s$.
Then for any $1\leq k\leq \lfloor\frac{h+r}{2}\rfloor\frac{q-1}{h}$,
there exists an $[[n,n-2k,k+1]]_q$-quantum MDS code.
\end{theorem}

\begin{proof}
Keep the notations as above.
We only need to prove the case when $r$ is even, since the case when $r$ is odd is completely similar.
Denote $r_1=r_2=t$, $\alpha=w^{m}$ and $\xi=w^{-\frac{q+1}{2}}$.
Let
$$B=\begin{pmatrix}
  \alpha^{2i_1}\xi^{2i_1} &  \alpha^{2i_2}\xi^{2i_2} & \dots & \alpha^{2i_{t}}\xi^{2i_{t}} \\
  \alpha^{4i_1}\xi^{2i_1} &  \alpha^{4i_2}\xi^{2i_2} & \dots & \alpha^{4i_{t}}\xi^{2i_{t}} \\
  \vdots & \vdots &  \ddots & \vdots\\
  \alpha^{2i_1(t-1)}\xi^{2i_1} &  \alpha^{2i_2(t-1)}\xi^{2i_2} & \dots & \alpha^{2i_{t}(t-1)}\xi^{2i_{t}} \\
\end{pmatrix},$$
$$C=\begin{pmatrix}
  \alpha^{2j_1+1}\xi^{2j_1+1} &  \alpha^{2j_2+1}\xi^{2j_2+1} & \dots & \alpha^{2j_{t}+1}\xi^{2j_{t}+1} \\
  \alpha^{4j_1+2}\xi^{2j_1+1} &  \alpha^{4j_2+2}\xi^{2j_2+1} & \dots & \alpha^{4j_{t}+2}\xi^{2j_{t}+1} \\
  \vdots & \vdots &  \ddots & \vdots\\
  \alpha^{(2j_1+1)(t-1)}\xi^{2j_1+1} &  \alpha^{(2j_2+1)(t-1)}\xi^{2j_2+1} & \dots & \alpha^{(2j_{t}+1)(t-1)}\xi^{2j_{t}+1} \\
\end{pmatrix}
\quad and\quad A=\begin{pmatrix}
B& C\\
B & -C\\
\end{pmatrix}
$$
be the matrices over $\F_{q^2}$, respectively.
Let
$$\tilde{A}=\begin{pmatrix}
B& 0\\
0 & C\\
\end{pmatrix}$$
be an $(2t-2)\times 2t$ matrix over $\F_{q^2}$.
It is easy to see that
equation $A\bm{u}^T=\bm{0}^T$ has the same solutions as equation $\tilde{A}\bm{u}^T=\bm{0}^T$.
Let $\tilde{A}_i\ (1\leq i\leq r)$ be the $(2t-2)\times (2t-1)$ matrix obtained from $\tilde{A}$ by deleting the $i$-th column.
Then $rank(\tilde{A})=rank(\tilde{A}_1)=\dots=rank(\tilde{A}_{r})=2t-2$.
By Lemma \ref{lem youjie xin},
equation $\tilde{A}\bm{u}^T=\bm{0}^T$
has a solution $\bm{u}=(u_1,u_2,\dots,u_{r})\in (\F_{q^2}^*)^{r}$.
Hence, equation $A\bm{u}^T=\bm{0}^T$
has a solution $\bm{u}=(u_1,u_2,\dots,u_{r})\in (\F_{q^2}^*)^{r}$.
Note that $\alpha^q=\alpha$, $\alpha^s=-1$ and $\xi^{q}=-\xi$,
then for any $1\leq i\leq h$ and $1\leq j\leq t-1$, we have
$$\alpha^{ijq}\xi^{iq}=(-\alpha^{j}\xi)^i=\alpha^{i(s+j)}\xi^{i}.$$
It follows that $A$ is row equivalent to $A^{(q)}$.
By Lemma \ref{lem youjie xin 2},
the equation $A\bm{u}^T=\bm{0}^T$
has a solution $\bm{u}=(u_1,u_2,\dots,u_{r})\in (\F_q^*)^{r}$.
Let $v_l^{q+1}=u_l$ for $1\leq l\leq r$.
Since $1\leq k\leq (s+t)\frac{q-1}{2s}$,
by Lemma \ref{lem zhenchu 3}, $qi+j+\frac{q+1}{2}=um$ if and only if $u\in \{1,2,\dots,t-1,s+1,s+2,\dots,s+t-1\}$.
Then by Eq. (\ref{eq qiuhe 5}), when $qi+j+\frac{q+1}{2}=um$ for $u \in \{1,2,\dots,t-1,s+1,s+2,\dots,s+t-1\}$,
we have
$$\langle \bm{a}^{qi+j},\bm{v}^{q+1}\rangle_E=
m\sum\limits_{l_1=1}^{t}\alpha^{2i_{l_1}u}\xi^{2i_{l_1}}u_{l_1}+
m\sum\limits_{l_2=1}^{t}\alpha^{(2j_{l_2}+1)u}\xi^{2j_{l_2}+1}u_{t+l_2}=0.$$
Then $\langle \bm{a}^{qi+j},\bm{v}^{q+1}\rangle_E=0$, for all $0\leq i,j\leq k-1$.
Therefore, by Lemma \ref{lem (a,v)=0},
the code $GRS_k(\bm{a},\bm{v})$ is Hermitian self-orthogonal.
Then the desired result follows from Lemma \ref{Lem zhuyao}.
This completes the proof.
\end{proof}

\begin{example}
In this example,
we give some quantum MDS codes from Theorem \ref{th4}.
\begin{itemize}
\item[(1)]
When $14\mid (q-1)$, let $(r,h)=(9,14)$ in Theorem \ref{th4}.
Then for any $1\leq k\leq \frac{11}{14}(q-1)$, there exists
a $[[\frac{9}{14}(q^2-1),\frac{9}{14}(q^2-1)-2k,k+1]]_q$-quantum MDS code;
\item[(2)]
When $16\mid (q-1)$, let $(r,h)=(10,16)$ in Theorem \ref{th4}.
Then for any $1\leq k\leq \frac{13}{16}(q-1)$, there exists
a $[[\frac{5}{8}(q^2-1),\frac{5}{8}(q^2-1)-2k,k+1]]_q$-quantum MDS code.
\end{itemize}
\end{example}

\subsection{Fifth construction of quantum MDS codes}

In \cite{RefJ (2018)W. two}, another necessary and sufficient condition for GRS codes to be Hermitian self-orthogonal codes is given.

\begin{lemma}(\cite{RefJ (2018)W. two})\label{lem GRS g(x)}
A codeword $\bm{c}=(v_1f(a_1),v_2f(a_2),\dots,v_nf(a_n))$ of $GRS_k(\bm{a},\bm{v})$
is contained in $GRS_{k}(\bm{a},\bm{v})^{\bot_H}$ if and only if there exists a polynomial $g(x)$
with $\deg(g(x))\leq n-k-1$,
such that
$$(v_1^{q+1}f^q(a_1), v_2^{q+1}f^q(a_2),\dots,v_{n}^{q+1}f^q(a_n))=(u_1^{-1}g(a_1),u_2^{-1}g(a_2),\dots,u_n^{-1}g(a_n)),$$
where $u_i=\prod_{1\leq j\leq n,j\neq i}(a_i-a_j)$, for $1\leq i\leq n$.
\end{lemma}

\begin{theorem}\label{th5}
Let $q$ be a prime power and $n=\frac{r(q^2-1)}{h}$, where $h\mid (q-1)$ and $1\leq r\leq h$. Then for any $1\leq k\leq r\frac{q-1}{h}$,
there exists an $[[n,n-2k,k+1]]_q$-quantum MDS code.
\end{theorem}

\begin{proof}
Denote $(\langle \theta\rangle)=(1,\theta,\dots,\theta^{m-1})$. Set
$$\bm{a}=(a_1,a_2,\dots,a_n)=(w^{i_1}(\langle \theta\rangle),w^{i_2}(\langle \theta\rangle),\dots,w^{i_r}(\langle \theta\rangle))\in \F_{q^2}^{r\frac{q^2-1}{h}},$$
where $i_1,i_2,\dots,i_r$ are distinct modulo $h$ and $n=rm$.
Note that $\prod_{0\leq l\leq m-1}(x-\theta^l)=x^{m}-1$ and
$\prod_{0\leq l\leq m-1,l\neq t}(x-\theta^l)=\sum_{i=0}^{m-1}x^i\theta^{t(m-1-i)}$,
it follows that
$$\prod_{0\leq l\leq m-1,l\neq t}(\theta^t-\theta^l)=\sum_{i=0}^{m-1}(\theta^t)^i\theta^{t(m-1-i)}=m\theta^{(m-1)t}.$$
Suppose $w^{i_\lambda}\langle \theta\rangle = A_\lambda$, then for each $a_i=w^{i_\lambda}\theta^t\in A_\lambda$, where $1\leq \lambda\leq r$ and $1\leq t\leq m$,
we have
\[\begin{split}
	       u_i&=\prod_{1\leq j\leq n, i\neq j}(a_i-a_j)\\
                  &=\prod_{x_\lambda\in A_\lambda,x_\lambda\neq a_i}(a_i-x_\lambda)\prod_{1\leq s\leq r, s\neq \lambda}\prod_{x_s\in A_s}(a_i-x_s)\\
                  &=\prod_{0\leq l\leq m-1,l\neq t}(w^{i_\lambda}\theta^t-w^{i_\lambda}\theta^l)
                  \prod_{1\leq s\leq r,s\neq \lambda}\prod_{0\leq l\leq m-1}(w^{i_\lambda}\theta^t-w^{i_s}\theta^l)\\
                   &=ma_i^{m-1}\prod_{1\leq s\leq r,s\neq \lambda}(w^{mi_\lambda}-w^{mi_s}).\\
	\end{split}\]
Note that $(q+1)\mid m$,
which implies that $\prod_{1\leq s\leq r,s\neq \lambda}(w^{mi_\lambda}-w^{mi_s})\in \F_q^*$, then $a_iu_i\in \F_q^*$.
It follows that $a_i^{-1}u_i^{-1}\in \F_q^*$,
then there exists $v_i'\in \F_{q^2}^*$ such that $(v_i')^{q+1}=a_i^{-1}u_i^{-1}$ for all $1\leq i\leq n$. Let $g(x)=(x^{q-1}+1)f^q(x)$.
If $1\leq k\leq r\frac{q-1}{h}$,
we have
$$\deg(g(x))\leq q-1+q(k-1)=qk-1\leq n-k-1.$$
It is easy to see that $a_i^{q}+a_i\in \F_{q}^*$,
then there is $b_i\in \F_{q^2}^*$ such that $b_i^{q+1}=a_i^{q}+a_i$. Let $v_i=v_i'b_i$,
we have
$$v_i^{q+1}f^q(a_i)=(v_i'b_i)^{q+1}f^q(a_i)=u_i^{-1}(a_i^{q-1}+1)f^q(a_i)=u_i^{-1}g(a_i),\ for\ all\ 1\leq i\leq n.$$
By Lemma \ref{lem GRS g(x)},
for any $1\leq k\leq r\frac{q-1}{h}$,
the code $GRS_k(\bm{a},\bm{v})$ is Hermitian self-orthogonal.
Then the desired result follows from Lemma \ref{Lem zhuyao}.
This completes the proof.
\end{proof}

\begin{remark}
\begin{itemize}
\item Compared with Theorems \ref{th3} and \ref{th4}, we add the case of $2\nmid h$;
\item The length of the code given in Theorem \ref{th5} appeared in \cite{RefJ (2017) X.Shi}, \cite{RefJ (2019)W. some} and \cite{RefJ (2021) application}. And the minimum distance of the codes we obtained can be larger than the previous conclusion by 1.
\end{itemize}
\end{remark}

\begin{example}
By taking $h=7$ in Theorem \ref{th5}.
Then when $7\mid (q-1)$, there exists a $[[\frac{r}{7}(q^2-1),\frac{r}{7}(q^2-1)-2k,k+1]]_q$-quantum MDS code
for any $1\leq k\leq \frac{r}{7}(q-1)$, where $1\leq r\leq 7$.
\end{example}

\section{Comparison}\label{sec5}

In this section, we make some detailed comparisons between our results and the previous results.

In Table \ref{tab:1}, we summarize the parameters of most previously known quantum MDS codes.
There are too many quantum MDS codes here, so we only list those whose minimum distance is greater than $q/2+1$ and are the best results.

In Table \ref{tab:3}, we list our constructions of quantum MDS codes.
From Tables \ref{tab:1} and \ref{tab:3}, it can be seen that the form of code length in Theorem \ref{th1} is the same as that in classes 10 and 11, the form of code length in Theorem \ref{th2} is the same as that in classes 12, 13, 14 and 15, and the form of code length in Theorems \ref{th3}, \ref{th4} and \ref{th5} is the same as that in classes 5, 6 and 7.
So we only need to compare these same length forms.
Now, we do some detailed comparisons in the following remark.

\begin{remark}(Comparison of Theorem \ref{th1})\label{remrak 1}
\begin{itemize}
\item[(1)] In [\cite{RefJ (2019)W. some} Theorem 4.3(i)] and [\cite{RefJ (2019)W. some} Theorem 5.3(i)], Fang et al. proved that there exists an $[[n=r\frac{q^2-1}{h}+1,n-2d,d]]$-quantum MDS code,
where $h\mid (q+1)$, $1\leq r\leq h$ and $2\leq d\leq (\lfloor\frac{h}{2}\rfloor+1)\frac{q+1}{h}$.
We can find that in Theorem \ref{th1}, the minimum distance can reach $(\frac{r+h-1}{2})\frac{q+1}{h}$.
Therefore, our results in Theorem \ref{th1} can reach a larger minimum distance;
\item[(2)]In [\cite{RefJ (2019)W. some} Theorem 4.3(ii)] and [\cite{RefJ (2019)W. some} Theorem 5.3(ii)], Fang et al. proved that there exists an $[[n=r\frac{q^2-1}{h}+1,n-2d,d]]$-quantum MDS code,
where $h\mid (q+1)$, $1\leq r< h$, $2\mid (r+h)$ and $2\leq d\leq (\frac{r+h}{2})\frac{q+1}{h}$.
We can find that in Theorem \ref{th1}, the condition is changed to $2\nmid (r+h)$.
Combining with them, the following corollary holds.
\end{itemize}
\end{remark}

\begin{corollary}
Let $n=r\frac{q^2-1}{h}+1$,
where $h\mid (q+1)$ and $1< r< \min\{q,h\}$.
Then for any $1\leq k\leq \lfloor \frac{r+h}{2}\rfloor\frac{q+1}{h}-1$,
there exists an $[[n,n-2k,k+1]]_q$-quantum MDS code.
\end{corollary}

\begin{remark}(Comparison of Theorem \ref{th2})\label{remrak 2}
\begin{itemize}
\item[(1)] In [\cite{RefJ (2019)W. some} Theorem 6.3], Fang et al. proved that there exists an $[[n=r\frac{q^2-1}{h},n-2d,d]]$-quantum MDS code,
where $2\mid h$, $h\mid (q+1)$, $1\leq r\leq h-1$, $2\nmid (r+h)$ and $2\leq d\leq (\frac{r+h-1}{2})\frac{q+1}{h}-1$.
We can find that in Theorem \ref{th2}, the minimum distance can reach $(\frac{r+h+1}{2})\frac{q+1}{h}-1$.
Therefore, our results in Theorem \ref{th2} can reach a larger minimum distance;
\item[(2)]In [\cite{RefJ (2017) X.Shi} Theorem 4.4], Shi et al. proved that there exists an $[[n=r\frac{q^2-1}{h},n-2d,d]]$-quantum MDS code,
where $2\nmid h$, $h\mid (q+1)$, $1\leq r\leq h-1$, $2\nmid (r+h)$ and $2\leq d\leq (\frac{r+h+1}{2})\frac{q+1}{h}$.
We can find that in Theorem \ref{th2}, the condition is changed to $2\mid h$.
Combining with them, the following corollary holds.
\end{itemize}
\end{remark}

\begin{corollary}\label{coro 1}
Let $n=r\frac{q^2-1}{h}$,
where $h\mid (q+1)$, $h\geq 3$ and $1\leq r <h$.
If $2\nmid (r+h)$,
then for any $1\leq k\leq (\frac{r+h+1}{2})\frac{q+1}{h}-2$,
there exists an $[[n,n-2k,k+1]]_q$-quantum MDS codes.
\end{corollary}

In Corollary \ref{coro 1}, take $r=h-1$, we get the following result.

\begin{corollary}
Let $n=(h-1)\frac{(q^2-1)}{h}$, where $h\mid (q+1)$ and $h\geq 3$.
Then for any $1\leq k\leq q-1$,
there exists an $[[n,n-2k,k+1]]_q$-quantum MDS codes.
\end{corollary}

\begin{remark}(Comparison of Theorems \ref{th3}, \ref{th4} and \ref{th5})\label{remrak 3}
\begin{itemize}
\item[(1)] In [\cite{RefJ (2017)Lem GRS} Theorem 3.2], Zhang et al. proved that there exists an $[[n=r\frac{q^2-1}{h},n-2d,d]]$-quantum MDS code,
where $2\mid h$, $h\mid (q-1)$, $1\leq r\leq h$, and $2\leq d\leq (\frac{h}{2}+1)\frac{q+1}{h}+1$.
We can find that in Theorems \ref{th3} and \ref{th4}, the minimum distance can reach $(\frac{h}{2}+r)\frac{q-1}{h}+1$ and $\lfloor\frac{h+r}{2}\rfloor\frac{q-1}{h}+1$, respectively.
Therefore, our results in Theorems \ref{th3} and \ref{th4} can reach a larger minimum distance;
\item[(2)]In [\cite{RefJ (2019)W. some} Theorem 3.2] and [\cite{RefJ (2017) X.Shi} Theorem 4.12], Fang and Shi et al. proved that there exists an $[[n=r\frac{q^2-1}{h},n-2d,d]]$-quantum MDS code,
where $h\mid (q-1)$, $1\leq r\leq h$ and $2\leq d\leq r\frac{q-1}{h}$.
We can find that $r\frac{q-1}{h}< (\frac{h}{2}+r)\frac{q-1}{h}+1$, $r\frac{q-1}{h}< \lfloor\frac{h+r}{2}\rfloor\frac{q-1}{h}+1$ and $r\frac{q-1}{h}< r\frac{q-1}{h}+1$.
Therefore, our results in Theorems \ref{th3}, \ref{th4} and \ref{th5} can reach a larger minimum distance.
\end{itemize}
\end{remark}

It is worth noting that Theorems \ref{th3}, \ref{th4} and \ref{th5} improve and generalize many of the previous conclusions (see Table \ref{tab:4}).

\newcommand{\tabincell}[2]{\begin{tabular}{@{}#1@{}}#2\end{tabular}}
\begin{table}
\caption{The conclusions generalized by Theorems \ref{th3}, \ref{th4} and \ref{th5}.}
\label{tab:4}
\begin{center}
\resizebox{\textwidth}{38mm}{
	\begin{tabular}{cccc}
		\hline
		Class& Length $n$ & Minimum Distance $d$ & References\\
        \hline
        1& $n=\lambda(q+1)$, $q$ odd, $\lambda$ odd, $\lambda\mid (q-1)$  & $2\leq d\leq \frac{q+1}{2}+\lambda$ & \cite{RefJ (2014) kai} \\
        \hline
        2&  $n=2\lambda(q+1)$, $q \equiv 1({\rm mod}\ 4)$, $\lambda$ odd, $\lambda\mid (q-1)$ & $2\leq d\leq \frac{q+1}{2}+2\lambda$ & \cite{RefJ (2014) kai}\\
        \hline
        3&  $n=2^fs(q+1)$, $2^e\parallel (q-1)$, $0\leq f<e$, $s\mid (q-1)$, $s$ odd & $2\leq d\leq \frac{q+1}{2}+2^fs$ & \cite{RefJ (2015) B.chen}\\
        \hline
       4& $n=\frac{q^2-1}{m}$, $m\mid (q-1)$, $m$ even & $2\leq d\leq \frac{q+1}{2}+\frac{q-1}{m}$ & \cite{RefJ (2016) X.He}\\
        \hline
        5&  $n=\frac{q-1}{2k+1}(q+1)$, $(2k+1)\mid (q-1)$, $(4k+1)\mid (q+1)$, $q$ odd & $2\leq d\leq \frac{q-1}{2}+\frac{q+1}{2(4k+1)}$ & \cite{RefJ (2016) X.He}\\
        \hline
         6&$n=(m_1+m_2-1)\frac{q^2-1}{2m_1m_2}$, odd $m_1<m_2$, $\gcd(m_1,m_2)=1$, $2m_1m_2= q-1$ & $2\leq d\leq \frac{q+1}{2}+m_1$ & \cite{RefJ (2016) X.He}\\
        \hline
        7&$n=\frac{q^2-1}{m_1}+\frac{q^2-1}{m_2}-q-1$, $2\mid m_1\mid (q-1)$, $2\mid m_2\mid (q-1)$, $lcm(m_1,m_2)=q-1$  & $2\leq d\leq \frac{q+1}{2}+\min\{\frac{q-1}{m_1},\frac{q-1}{m_2}\}$ & \cite{RefJ (2016) X.He}\\
        \hline
        8& $n=bm(q+1)$, $2m\mid (q-1)$, $bm\leq q-1$ & $2\leq d\leq \frac{q+1}{2}+m$ & \cite{RefJ (2017)Lem GRS}\\
        \hline
         9& $n=(bm+c(m-1))(q+1)$, $2m\mid (q-1)$, $b,c\geq 0$, $(b+c)m\leq q-1$ and $b\geq 1$ or $m\geq 2$  & $2\leq d\leq \frac{q+1}{2}+m$ & \cite{RefJ (2017)Lem GRS}\\
        \hline
        10 &$n=c(q+1)$, $q=2am+1$, $\gcd(a,m)=1$, $1\leq c\leq a+m-1\leq am$ & $2\leq d\leq \frac{q+1}{2}+c$ & \cite{RefJ (2017)Lem GRS}\\
        \hline
        11 &$n=c(q+1)$, $q=2am+1$, $\gcd(a,m)=1$, $a+m\leq c\leq 2(a+m-1)$ & $2\leq d\leq \frac{q+1}{2}+\lfloor \frac{c}{2}\rfloor$ & \cite{RefJ (2017)Lem GRS}\\
        \hline
        12 &$n=(t+1)\frac{q^2-1}{h}$, $q-1=mh$, $h>1$, $m>1$, $1\leq t\leq h-1$ & $2\leq d\leq (t+1)\frac{q-1}{h}$& \cite{RefJ (2017) X.Shi}\\
        \hline
        13 &$n=(q-\delta-1)(q+1)$, $q$ even, $0\leq \delta\leq q-3$ & $2\leq d\leq q-\delta-1$& \cite{RefJ (2018) X.Shi CC}\\
        \hline
        14 &$n=(q-1-2\delta)(q+1)$, $q\equiv 3($mod $4)$, $q>3$, $0\leq \delta\leq \frac{q-5}{2}$ & $2\leq d\leq q-2-2\delta$& \cite{RefJ (2018) X.Shi CC}\\
        \hline
        15 &$n=s(q+1)$, $1\leq s\leq q-1$ & $2\leq d\leq s$& \cite{RefJ (2021) application}\\
        \hline
	\end{tabular}}
\end{center}
\end{table}

\begin{table}
\caption{Some known results of $[[n,n-2d+2,d]]_q$-quantum MDS codes}
\label{tab:1}
\begin{center}
\resizebox{\textwidth}{65mm}{
	\begin{tabular}{cccc}
		\hline
		Class  & Length $n$ & Minimum Distance $d$ & References\\
		\hline
		1  &  $n\leq q+1$ & $2\leq d\leq \frac{n}{2}+1$ &  \cite{RefJ (2004) n<q+1(1),RefJ (2004) n<q+1(2)} \\
        \hline
        2 & $n=q^2+1$ & $2\leq d\leq q+1$, $d\neq q$ & \cite{RefJ (2011) n=q^2+1(1),RefJ (2015) n=q^2+1(2),RefJ (2013) n=q^2+1(3),RefJ (2008) n=q^2+1(4),RefJ (2021) n=q^2+1(5)} \\
        \hline
        3 & $n=r(q+1)+2$, $1\leq r\leq q-1$ & $2\leq d\leq r+2$, $(p,r,d)\neq(2,q-1,q)$& \cite{RefJ (2018)W. two} \\
        \hline
        4 &  $n=r\frac{q^2-1}{h}+1$, $h|(q-1)$, $1\leq r\leq h$ & $2\leq d\leq r\frac{q-1}{h}+1$ & \cite{RefJ (2019)W. some} \\
        \hline
        5 & $n=r\frac{q^2-1}{h}$, $h|(q-1)$, $1\leq r\leq h$ & $2\leq d\leq r\frac{q-1}{h}$ & \cite{RefJ (2017) X.Shi,RefJ (2018) X.Shi CC,RefJ (2019)W. some} \\
        \hline
        6 &$n=r\frac{q^2-1}{h}$, $2\mid h$, $h|(q-1)$, $1\leq r\leq h$ & $2\leq d\leq (\frac{h}{2}+1)\frac{q-1}{h}+1$ & \cite{RefJ (2014) kai,RefJ (2015) B.chen,RefJ (2016) X.He,RefJ (2017)Lem GRS} \\
        \hline
        7 & $n=\frac{q^2-1}{2}$, $q$ odd & $2\leq d\leq q$& \cite{RefJ (2014) kai} \\
        \hline
        8 & $n=q^2$ & $2\leq d\leq q$& \cite{RefJ (2014) L.jin,RefJ (2008) n=q^2+1(4)} \\
        \hline
        9 & $n=tq$, $1\leq t\leq q$ & $2\leq d\leq \lfloor  \frac{tq+q-1}{q+1} \rfloor+1$& \cite{RefJ (2008) n=q^2+1(4),RefJ (2018)W. two} \\
        \hline
        10 & $n=r\frac{q^2-1}{h}+1$, $h|(q+1)$, $1\leq r\leq h$ & $2\leq d\leq (\lfloor \frac{h}{2}\rfloor+1)\frac{q+1}{h}$ & \cite{RefJ (2016) X.He,RefJ (2017) L.Jin,RefJ (2019)W. some} \\
        \hline
        11 &$n=r\frac{q^2-1}{h}+1$, $h|(q+1)$, $1\leq r\leq h-2$, $2|(h+r)$ & $2\leq d\leq (\frac{h+r}{2})\frac{q+1}{h}$ & \cite{RefJ (2014) L.jin,RefJ (2019)W. some}  \\
        \hline
        12 & $n=r\frac{q^2-1}{h}$, $h|(q+1)$, $1\leq r< h$, $2\nmid (h+r)$, $2\nmid h$ &  $2\leq d\leq (\frac{h+r+1}{2})\frac{q+1}{h}-1$& \cite{RefJ (2015) L.W,RefJ (2015) B.chen,RefJ (2017) X.Shi,RefJ (2017) L.Jin}\\
        \hline
        13 & $n=r\frac{q^2-1}{h}$, $h|(q+1)$, $1\leq r< h$, $2\nmid (h+r)$, $2\mid h$ &  $2\leq d\leq (\frac{h+r-1}{2})\frac{q+1}{h}-1$ & \cite{RefJ (2016) X.He,RefJ (2019)W. some}\\
        \hline
        14 &$n=r\frac{q^2-1}{h}$, $h|(q+1)$, $1\leq r\leq h-2$, $2|(h+r)$ & $2\leq d\leq (\frac{h+r}{2})\frac{q+1}{h}-1$ & \cite{RefJ (2017) X.Shi,RefJ (2019)W. some} \\
        \hline
        15 &$n=t(q-1)$, $1\leq t\leq q-1$ & $2\leq d\leq \lfloor \frac{tq-1}{q+1}\rfloor+1$ & \cite{RefJ (2021) application}  \\
        \hline
        16 & $n=\frac{q^2+1}{2}$, $q$ odd &  $2\leq d\leq q$, $d$ odd& \cite{RefJ (2014) L.jin,RefJ (2013) n=q^2+1(3)} \\
        \hline
        17 &$n=\frac{q^2+1}{5}$, $q\equiv \pm 3({\rm mod}\ 10)$ & $2\leq d\leq \frac{3q\pm 1}{5}$, $d$ even & \cite{RefJ (2014) kai,RefJ (2016) L.Hu,RefJ (2015) T.Zhang} \\
        \hline
        18 &$n=\frac{q^2+1}{5}$,  $q\equiv \pm 2({\rm mod}\ 10)$ & $2\leq d\leq \frac{3q\mp 1}{5}$, $d$ odd & \cite{RefJ (2016) S.Li} \\
        \hline
        19 & \tabincell{c}{$n=r\frac{q^2-1}{s}+l\frac{q^2-1}{t}-rl\frac{q^2-1}{st}$, odd $s\mid (q+1)$,\\ even $t\mid (q-1)$, $r\leq s-1$, $l\leq t$, $rl\frac{q^2-1}{st}< q-1$ }
        & $2\leq d\leq \min\{\lfloor \frac{s+r}{2}\rfloor\frac{q+1}{s}-1, (\frac{t}{2}+1)\frac{q-1}{t}+1\}$ & \cite{RefJ (2020) X.Fang} \\
        \hline
        20 & \tabincell{c}{$n=r\frac{q^2-1}{s}+l\frac{q^2-1}{t}-rl\frac{q^2-1}{st}+1$, odd $s\mid (q+1)$,\\ even $t\mid (q-1)$, odd $r\leq s-1$, $l\leq t$, $rl\frac{q^2-1}{st}< q-1$ }
        & $2\leq d\leq \min\{ \frac{s+r}{2}\frac{q+1}{s}, (\frac{t}{2}+1)\frac{q-1}{t}+1\}$ & \cite{RefJ (2020) X.Fang} \\
        \hline
        21 & \tabincell{c}{$n=r\frac{q^2-1}{s}+l\frac{q^2-1}{t}$, even $s\mid (q+1)$,\\ even $t\mid (q-1)$, $r\leq \frac{s}{2}$, $l\leq \frac{t}{2}$}
        & $2\leq d\leq \min\{ \lfloor \frac{s+r}{2}\rfloor\frac{q+1}{s}-1, (\frac{t}{2}+1)\frac{q-1}{t}+1\}$ & \cite{RefJ (2020) X.Fang} \\
        \hline
	\end{tabular}}
 \begin{tablenotes}
     \footnotesize
    \item Note: in \cite{RefJ (2016) X.He}, \cite{RefJ (2017)Lem GRS}, \cite{RefJ (2019) F.Tian} and \cite{RefJ (2021) Ball determine} many quantum MDS codes were also introduced.
    \end{tablenotes}
\end{center}
\end{table}

\begin{table}
\caption{Our new constructions of $[[n,n-2d+2,d]]_q$-quantum MDS codes}
\label{tab:3}
\begin{center}
\resizebox{\textwidth}{25mm}{
	\begin{tabular}{cccc}
		\hline
		Forms of $n$  & Length $n$ & Minimum Distance $d$ & References\\
        \hline
        $(q-1)\mid (n-1)$& \tabincell{c}{$n=r\frac{q^2-1}{h}+1$, $h\mid (q+1)$,\\ $1< r< \min\{q,h\}$, $2\nmid (r+h)$} & $2\leq d\leq (\frac{r+h-1}{2})\frac{q+1}{h}$ & Theorem \ref{th1} \\
        \hline
        $(q-1)\mid n$ & \tabincell{c}{$n=r\frac{q^2-1}{h}$, $h\mid (q+1)$,\\ $1\leq r\leq h$, $2\nmid (r+h)$, $2\mid h$} & $2\leq d\leq (\frac{r+h+1}{2})\frac{q+1}{h}-1$ & Theorem \ref{th2}\\
        \hline
        $(q+1)\mid n$ & $n=r\frac{q^2-1}{h}$, $2\mid h$, $h\mid (q-1)$, $1\leq r\leq \frac{h}{2}$ & $2\leq d\leq (\frac{h}{2}+r)\frac{q-1}{h}+1$ & Theorem \ref{th3}\\
        \hline
        $(q+1)\mid n$ &  $n=r\frac{q^2-1}{h}$, $2\mid h$, $h\mid (q-1)$, $\frac{h}{2}<r< h$ & $2\leq d\leq \lfloor\frac{h+r}{2}\rfloor\frac{q-1}{h}+1$ & Theorem \ref{th4}\\
        \hline
        $(q+1)\mid n$ &  $n=r\frac{q^2-1}{h}$, $2\nmid h$, $h\mid (q-1)$, $1<r< h$ & $2\leq d\leq r\frac{q-1}{h}+1$ & Theorem \ref{th5}\\
        \hline
	\end{tabular}}
\end{center}
\end{table}

\section{Conclusions}\label{sec6}

In this paper,
we construct five new classes of $q$-ary quantum MDS codes via Hermitian self-orthogonal GRS codes
(see Theorems \ref{th1}, \ref{th2}, \ref{th3}, \ref{th4}, \ref{th5}).
It turns out that the quantum MDS codes we have constructed are new, since the parameters of these codes cannot be derived from previous conclusions
(see Remarks \ref{remrak 1}, \ref{remrak 2}, \ref{remrak 3}).
The minimum distance of all the $q$-ary quantum MDS codes constructed in this paper can be larger than $q/2+1$.

\section*{Acknowledgments}

This research was supported by the National Natural Science Foundation of China (No.U21A20428 and 12171134).

\end{document}